\pdfoutput=1
\documentclass{article}
\usepackage[usenames, dvipsnames]{xcolor}
\usepackage[utf8]{inputenc}
\usepackage[english]{babel}
\usepackage{amsmath,amsfonts}
\usepackage{amssymb}
\usepackage[textwidth=2in]{todonotes}
\usepackage{mathtools}
\usepackage{amsthm}
\usepackage{eufrak}
\usepackage{moresize}
\usepackage{tikz}
\usepackage[letterpaper, portrait, left=1.3in, right=1.3in, top=1.0in, bottom=1.2in]{geometry}
\usepackage[numbers,sectionbib]{natbib}
\usepackage[hyperindex,breaklinks]{hyperref}
\usepackage[ruled,vlined, linesnumbered]{algorithm2e}
\usepackage{cleveref}
\usepackage[shortlabels]{enumitem}
\usepackage{booktabs} 
\usepackage{xparse} 
\usepackage{array}
\usepackage{booktabs}
\usepackage{pifont}
\usepackage[toc]{appendix}
\usepackage{autonum}
\usepackage{makecell}
\usepackage{float}
\usepackage{tikz}

\makeatletter
\newcommand\incircbin
{%
  \mathpalette\@incircbin
}
\newcommand\@incircbin[2]
{%
  \mathbin%
  {%
    \ooalign{\hidewidth$#1#2$\hidewidth\crcr$#1\bigcirc$}%
  }%
}

\makeatother

\SetKwRepeat{Repeat}{repeat}{}

\usetikzlibrary{patterns}

\def\squareforqed{\leavevmode\hbox to.77778em{\hfil\vrule\vbox to.675em{\hrule width.6em\vfil\hrule}\vrule\hfil}} 

\newtheorem{theorem}{Theorem}

\newtheorem{definition}[theorem]{Definition}
\newtheorem{lemma}[theorem]{Lemma}

\newtheorem{observation}[theorem]{Observation}

\theoremstyle{remark}

\SetKwFor{RepTimes}{repeat}{times}{end}

\newcommand{\nosemic}{\renewcommand{\@endalgocfline}{\relax}}
\newcommand{\dosemic}{\renewcommand{\@endalgocfline}{\algocf@endline}}

\date{}
\title{Sampling an Edge Uniformly in Sublinear Time}
\author{Jakub Tětek\\\texttt{j.tetek@gmail.com}\\Basic Algorithms Research Copenhagen\\University of Copenhagen}
\begin{document}
\maketitle
\begin{abstract}
The area of sublinear algorithms have recently received a lot of attention. In this setting, one has to choose specific access model for the input, as the algorithm does not have time to pre-process or even to see the whole input. A fundamental question remained open on the relationship between the two common models for graphs -- with and without access to the ``random edge" query -- namely whether it is possible to sample an edge uniformly at random in the model without access to the random edge queries.

In this paper, we answer this question positively. Specifically, we give an algorithm solving this problem that runs in expected time $O(\frac{n}{\sqrt{m}} \log n)$. This is only a logarithmic factor slower than the lower bound given in \cite{original_paper}. Our algorithm uses the algorithm from \cite{bounded_arboricity} which we analyze in a more careful way, leading to better bounds in general graphs. We also show a way to sample edges $\epsilon$-close to uniform in expected time $O(\frac{n}{\sqrt{m}} \log \frac{1}{\epsilon})$, improving upon the best previously known algorithm.

We also note that sampling edges from a distribution sufficiently close to uniform is sufficient to be able to simulate sublinear algorithms that use the random edge queries while decreasing the success probability of the algorithm only by $o(1)$. This allows for a much simpler algorithm that can be used to emulate random edge queries.
\end{abstract}

\section{Introduction}
The rise in huge datasets has sparked interest in the design of sublinear algorithms, where the goal is to approximately compute some function of the data while using sublinear amount of resources. The algorithms with sublinear space complexity -- streaming algorithms -- are well-studied \cite{McGregor2014, Muthukrishnan2005} and algorithms with sublinear time complexity have also recently received signigicant amount of interest \cite{Ron2019}. An interesting problem in such algorithms is that they cannot even look at the whole dataset but instead look at some random sample. This means that it has to be decided how exactly is one allowed to sample and query the input.

Different models for sublinear computation on graphs have been considered and give rise to algorithms with different running times. For example, in \cite{Gonen2011}, the authors show a lower bound on the problem of counting stars in a graph, while in \cite{Aliakbarpour2016}, the authors show an algorithm which beats this lower bound by using a stronger model. We consider the \textit{standard sublinear graph model}, defined as being able to perform uniform random vertex queries, degree queries, neighborhood queries. We also consider the \textit{extended sublinear graph model} which in addition to the mentioned queries also has the added operation of picking a random edge. For further discussion and models that allow pair queries (which we also consider in this work), see \Cref{sec:preliminaries}.

It is a fundamental question how these two models relate to each other. First progress in this direction has been made by \citet{original_paper}. In their paper, the authors show an algorithm that samples an edge approximately uniformly and can be used to (approximately) simulate the ``random edge" operation in time $O(\frac{n}{\sqrt{\epsilon m}})$. They also show a simple lower bound of $\Omega(\frac{n}{\sqrt{m}})$. While their algorithm is close to optimal for approximate sampling, they do not investigate in general the effects of \textit{inexact} sampling on algorithms wishing to use it to emulate the extended sublinear graph model. In this paper, the authors also pose a question whether it is possible to sample an edge exactly uniformly in sublinear time in the standard graph model. We answer this question positively by giving such algorithm that runs in expected $O(\frac{n}{\sqrt{m}} \log \frac{1}{n})$ time, thus resolving (up to a logarithmic factor) the question of relation of the two models, at least when only a constant number of random edge queries are desired. Moreover, we show that for small enough $\epsilon$, the difference between the uniform and $\epsilon$-close to uniform (defined in \Cref{sec:preliminaries}) distibutions is immaterial, allowing us to give a more practical algorithm for simulation of random edge queries with better concentration of running time.

\section{Preliminaries} \label{sec:preliminaries}
\begin{definition}
A probability distribution $\mathcal{P}$ is said to be pointwise $\epsilon$-close to $\mathcal{Q}$, denoted $|\mathcal{P} - \mathcal{Q}|_P \leq \epsilon$, if
\[
|\mathcal{P}(x)-\mathcal{Q}(x)| \leq \epsilon \mathcal{Q}(x), \quad \text { or equivalently } \quad 1-\epsilon \leq \frac{\mathcal{P}(x)}{\mathcal{Q}(x)} \leq 1+\epsilon
\]
\end{definition}
\noindent
In this paper, we consider distributions pointwise $\epsilon$-close to uniform. This measure of similarity of distributions is related to the total variational distance, the definition of which we recall below. Specifically, for any $\mathcal{P}, \mathcal{Q}$, it holds that $|\mathcal{P}- \mathcal{Q}|_{TV} \leq |\mathcal{P}- \mathcal{Q}|_{P}$ \cite{original_paper}.
\begin{definition}
The total variational distance between $\mathcal{P}$ and $\mathcal{Q}$ is defined as
\[
|\mathcal{P}- \mathcal{Q}|_{TV} = \frac{1}{2}\sum_{\omega \in \Omega} |\mathcal{P}(\omega)- \mathcal{Q}(\omega)|
\]
\end{definition}
\noindent
We say a probability distribution $\mathcal{P}$ is $\epsilon$-close to $\mathcal{Q}$ (not pointwise) if $|\mathcal{P} - \mathcal{Q}|_{TV} \leq \epsilon$.

In this paper, we replace each edge by two directed edges going in the oposite directions. Edges of the original graph can then be sampled by directing them in an arbitrary way and then forgetting the direction of a sampled directed edge.

We use the following notion of heavy and light edges, as defined in \cite{original_paper}. Let $\theta = \lceil\sqrt{2m}\rceil$. We call a vertex $v$ light if $d(v) \leq \theta$ and heavy otherwise. An edge $uv$ is called light (heavy) when the vertex $u$ is light (heavy). We denote $d_H(v)$ the number of heavy neighbors of $v$.

In this paper, we assume the knowledge of the number of edges. If we do not assume this, we might set $m$ to any upper bound. Indeed, setting $m$ to any number greater than the number of edges in the graph would work, but higher values of $m$ would result in worse running times. If one does not know the number of edges exactly, this can be resolved by using the algorithm from \cite{Goldreich2008}, to get a constant factor approximation with probabiity at least $1-\delta$ in time $O(\frac{n}{\sqrt{m}} \log \frac{1}{\delta})$.

\section{Our results}
In \Cref{sec:approx_sampling}, we show that the algorithm from \cite{bounded_arboricity}, for modified values of its parameters, performs considerably better than was previously known. Namely, that it can sample edges pointwise $\epsilon$-close to uniform in time $O(\frac{n}{\sqrt{m}} \log \frac{1}{\epsilon})$. The best previously known algorithm runs in time $O(\frac{n}{\sqrt{\epsilon m}})$\cite{original_paper}. This new analysis is the main contribution of this paper.

It is this speedup which allows us to use this result for efficiently sampling in expected time $O(\frac{n}{\sqrt{m}} \log \frac{1}{n})$ from the \textit{exact uniform distribution} (\Cref{sec:exact_sampling}) as well as to argue that for small enough $\epsilon$, the difference between sampling from the uniform distribution and sampling pointwise $\epsilon$-close to uniform is immaterial (\Cref{sec:doesnt_matter}). Polynomially small $\epsilon$ is sufficient, leading to same expected time complexity as for exact sampling but with better concentrtion of the running time, making this approach especially suitable for use in practice.

\subsection{Techniques}
We now give brief overview of the approximate sampling, as this is the algorithm on which the results from the rest of the paper build. The algorithm has previously been described in \cite{bounded_arboricity} in the context of bounded arboricity graphs. In this paper, we use this algorithm but set its parameters differently.

The algorithm is based on constrained random walks of length chosen at random; the algorithm returns the last edge of the walk. The random walk has constraints, which, when not satisfied, cause the algorithm to fail and restart. These constraints are (1) the first vertex $v$ of the walk is light and all subsequent vertices are heavy and (2) picking $X \sim Bern(d(v)/\sqrt{2m})$, the first step of the walk fails if $X = 0$ (note that this is equivalent to using rejection sampling to sample the first edge od the walk).

The intuition the following. It is easy to sample light edges -- one may pick vertex at random, choose $j$ uniformly at random from $[\theta]$ and return the $j$th outgoing edge incident to the picked vertex or repeat if there is less than $j$ neighbors. This corresponds to a random walk of length 1.

When we perform a random walk of length 2 subject to the constraints, the probability that a specified heavy vertex is the middle vertex of the walk is proportional to the number of its light neighbors. Since $\theta$ is set such that at least a half of heave vertex's neighbors are light, this is at least half of its degree and we are, therefore, able to sample edges pointwise $2$-approximately uniformly.

It can be shown that by further increasing the maximum allowed length (recall that the length is chosen at random from $[k]$ for some $k$) of the random walk, the accuracy of sampling increases exponentially.

In \Cref{sec:exact_sampling}, we then show that if the accuracy of the approximate sampling is sufficient, one may correct the difference to the uniform distribution while contributing little to the expected running time. In \Cref{sec:doesnt_matter}, we show a simple coupling argument, proving that if an algorithm succeeds with probability $c$ in the extended sublinear graph model, it succeeds with probability $c - o(1)$ when uniform random edge query is emulated by \texttt{sample\_edge} with $\epsilon$ being polynomially small in the standard sublinear graph model.

\section{Related results}
Random edge queries have been used by \citet{Gonen2011} for counting stars in a time beating the lower bound shown by \citet{Aliakbarpour2016}. This work has been generalized to counting arbitrary subgraphs in \cite{Assadi2018}.

The problem we address in this paper was first considered in \cite{original_paper}. As mentioned above, the authors show an algortihm for pointwise $\epsilon$-approximate sampling that runs in $\frac{n}{\sqrt{\epsilon m}}$ and show a lower bound of $\frac{n}{\sqrt{m}}$. They also pose the open problem of exact sampling, which is solved in this paper. This problem in the context of bounded arboricity graphs have been considered in \cite{bounded_arboricity}. From this paper also comes, up to a difference in parameters, what we call \Cref{alg:sample_edge}. The problem of sampling multiple edges has been considered in \cite{Eden2020}.

We mention two more papers that consider edge sampling. In their paper, \citet{Eden2019} show an algorithm for edge sampling parameterized by the arboricity (or, oeuivalently, the degeneracy) of the input graph. The time complexity of this algorithm is in the worst case, up to $\text{poly}(\log n, \epsilon^{-1})$ factor the same as that in this paper and \cite{original_paper} but can be significantly better when the arboricity of the graph is small. In \cite{Bishnu2019}, the authors solve the problem of weighted sampling of edges using a stronger inner product oraculum. Using this algorithm for the unweighted case results in a time complexity same as that from \cite{original_paper}.

Sampling vertices in a more restrictive local query model was studied in \cite{Chierichetti2018}. In this model, we only have access to some initial seed vertex and to neighbor queries. The algorithm given in this paper is based on random walks and its running time is parameterized by the mixing time of the graph.
\section{Approximate sampling} \label{sec:approx_sampling}
In this section, we show the \Cref{alg:sample_edge} which samples an edge pointwise $\epsilon$-approximately in expeced time $O(\frac{n}{\sqrt{m}}\log \frac{1}{\epsilon})$. This algorithm is, up to the change of parameters, the one used in \cite{bounded_arboricity}. We provide a tighter analysis in this paper. This algorithm works by repeated sampling attempts, each succeeding with probability $\frac{\sqrt{m}}{n\log \epsilon^{-1}}$. We then show that upon successfuly sampling an edge, the distribution is pointwise $\epsilon$-close to uniform.

\begin{figure}
\centering
\begin{minipage}[t]{0.495\textwidth}
\vspace{0pt}
\begin{algorithm}[H]
Sample a vertex $u_0 \in V$ uniformly at random and query for its degree.\\
If $d(u_0)>\theta$ return fail.\\
Choose a number $j \in[\theta]$ uniformly at random.\\
Let $u_1$ be the $j^{\text {th }}$ neighbor of $u_0$\\
\For{$i = 2, \cdots, k$}{
If no vertex was returned or if the returned vertex is light then return fail\\
Set $u_i$ to  a random neighbor of $u_{i-1}$.\\
}
Return $(u_{k-1}, u_k)$.

\caption{\texttt{Sampling\_attempt($k$)}} \label{alg:sampling_attempt}
\end{algorithm}
\end{minipage}
\begin{minipage}[t]{0.495\textwidth}
\vspace{0pt}
\begin{algorithm}[H]
Pick $k$ from $\{1, \cdots, \ell = \lceil \log_2 \frac{1}{\epsilon} \rceil + 1\}$ uniformly at random\\
Call \texttt{Sampling\_attempt(k)}, if it fails, go back to line 1, otherwise return the result

\caption{\texttt{Sample\_edge($\epsilon$)}} \label{alg:sample_edge}
\end{algorithm}
\end{minipage}
\end{figure}

We show separately for light and heavy edges that they are sampled almost uniformly. The case of light edges (\Cref{obs:light_uniform}) is analogous to a proof in \cite{original_paper}. We include it here for completeness.

\begin{observation} \label{obs:light_uniform}
Any fixed light edge $e=(u,v)$ is chosen by \Cref{alg:sampling_attempt} with probability $\frac{1}{\ell n \theta}$
\end{observation}
\begin{proof}
The edge $uv$ is chosen exactly when $k=1$ (happens with probability $\frac{1}{\ell}$), $u_0 = u$ (happens with probability $\frac{1}{n}$), and $j$ is such that the $j^{\text{th}}$ neighbor of $u$ is $v$ (happens with probability $\frac{1}{\theta}$). This gives total probability of $\frac{1}{\ell n \theta}$.
\end{proof}

We now analyze the case of heavy edges. Before that, we define for $v$ being a heavy vertex $h_{v,1} = \frac{d_H(v)}{d(v)}$ and for $i \geq 2$
\[
h_{v,i}= h_{v,1} \sum_{w \in N_h(v)} h_{w,i-1} / d_H(v)
\]
and for $v$ being a light vertex and $i$ being any natural number $h_{v,i} = 0$. We use the $h$-values for light vertices in the next section.

\begin{lemma} \label{lem:heavy_prob}
For any heavy edge $(v,w)$
\[
P[u_{k-1} = v, u_k = w | k \geq 2] = (1- h_{v,\ell})\frac{1}{(\ell-1) n \theta}
\]
\end{lemma}

\begin{proof}
$k$ is chosen uniformly at random from $\{2,\cdots,r\}$. We show by induction on $r$ that in this setting, it holds that $P[u_{k-1} = v | r] = (1- h_{v,r})\frac{d(v)}{(r-1) n \theta}$. If we show this, the lemma follows by substituting $r = \ell$ and by unifomity of $u_\ell$ on the neighborhood of $u_{\ell-1}$, thus decreasing the probability by a factor of $d(v)$.

For $r=2$, the claim holds because when $k=2$, there is probability $\frac{1}{\ell n \theta}$ that we come to $v$ from any of the $(1-h_{v,1}) d(v)$ adjacent light vertices. 

We now show the induction step. Consider the vertices $w \in N(V)$ and consider $P[u_{k-2} = w | r]$. Since $k$ is chosen uniformly from $\{2, \cdots, r\}$, this is the same as $P[u_{k-1} = w]$ for $k$ chosen uniformly from $\{1, \cdots, r-1\}$. Now 
\begin{align}
P[u_{k-1} = v | r] =& \sum_{w \in N(v)} P[u_{k-2} = w|r] P[u_{k-1}=v | u_{k-2}=w] \\
=& \sum_{w \in N_L(v)} P[u_0 = w] P[u_1 = v | u_0 = w] P[k = 2] \\&+ \sum_{w \in N_H(v)} P[u_{k-1} = w | r-1] P[u_{k} = v | u_{k-1} = w] P[k > 2] \\
=& \sum_{w \in N_L(v)} \frac{1}{n} \frac{1}{\theta} \frac{1}{r-1} + \sum_{w \in N_H(v)} (1- h_{w,r-1})\frac{d(w)}{(r-2) n \theta} \frac{1}{d(w)} \frac{r-2}{r-1} \\
=& (1-h_{v,1}) \frac{d(v)}{(r-1) n \theta} + \sum_{w \in N_H(v)} (1- h_{w,r-1})\frac{1}{(r-1) n \theta} \\
=& (1-h_{v,1}) \frac{d(v)}{(r-1) n \theta} + h_{v,1} d(v) (1- \sum_{w \in N_h(v)}h_{w,r-1}/d_H(v))\frac{1}{(r-1) n \theta} \\
=& (1-h_{v,1}) \frac{d(v)}{(r-1) n \theta} + (h_{v,1} - h_{v,r})\frac{d(v)}{(r-1) n \theta} \\
=& (1 - h_{v,r})\frac{d(v)}{(r-1) n \theta}
\end{align}
\end{proof}

Before putting it all together, we will need the following bound on $h_{v,i}$.
\begin{lemma} \label{lem:h_bound}
For any $v \in V_H(G)$ and $k \geq 1$ it holds that
\[
h_{v,k} \leq 2^{-k}
\]
\end{lemma}
\begin{proof}
We first prove that for any $w \in V(G)$, it holds that $h_{v,1} \leq 1/2$. This has been shown in \cite{original_paper} and we include this for completeness. We then argue by induction that this implies the lemma.

Since $v$ is heavy, it has more than $\theta$ neighbors. Moreover, there can be at most $\frac{m}{\theta}$ heavy vertices, meaning that the fraction of heavy neighbors of $v$ can be bounded as folows
\[
h_{v,1} \leq \frac{m}{\theta} \frac{1}{\theta} = \frac{m}{\lceil\sqrt{2m}\rceil} \frac{1}{\lceil\sqrt{2m}\rceil} \leq \frac{1}{2}
\]

We now show the claim by induction. We have shown the base case and it therefore remains to prove the induction step:
\[
h_{v,i}= h_{v,1} \sum_{w \in N_h(v)} h_{w,i-1} / d_H(v) \leq \frac{1}{2} \sum_{w \in N_h(v)} 2^{-(i-1)} / d_H(v) = 2^{-i}
\]
\end{proof}

We can now prove the following theorem
\begin{theorem}
For $\epsilon \leq \frac{1}{2}$, the \Cref{alg:sample_edge} has expected running time of $O(\frac{n}{\sqrt{m}} \log \frac{1}{\epsilon})$ and samples an edge from distribution pointwise $\epsilon$-close to uniform.
\end{theorem}
\begin{proof}
We first show that in an iteration of \Cref{alg:sample_edge}, each edge is sampled with probability in $[(1- \epsilon)\frac{1}{\ell n \theta}, \frac{1}{\ell n \theta}]$. 

First, let $(v,w)$ be a light edge. Then it is chosen exactly when $k=1$, $u_0 = v$ and $w$ is the $j^{\text {th }}$ neighbor of $v$
\[
P[k=1, u_0 = v, w = u_1] = P[k=1] P[u_0 = v|k=1] P[w = u_1 | u_0 = v, k=1] = \frac{1}{\ell n \theta}
\]

A heavy edge $(v,w)$ is chosen when $k \geq 2$, $u_{k-1} = v$ and $w = u_k$. Using \Cref{lem:heavy_prob},
\begin{align}
P[k \geq 2, u_{k-1} = v, u_k = w] &= P[k \geq 2]P[u_{k-1} = v, u_k = w | k \geq 2] \\
&= \frac{\ell - 1}{\ell} (1- h_{v,\ell})\frac{1}{(\ell-1) n \theta} \\
&= (1- h_{v,\ell})\frac{1}{\ell n \theta}\\
&\geq (1- 2^{-\ell})\frac{1}{\ell n \theta} \geq (1-\tfrac{1}{2}\epsilon) \frac{1}{\ell n \theta}
\end{align}
where the inequality holds by \Cref{lem:h_bound}.

\paragraph{Approximate uniformity.}Let $e$ be the sampled edge and $e', e''$ some fixed edges. Then since $P[e = e' | \text{success}] = \frac{P[e = e']}{P[\text{success}]}$ and for any fixed $e'$ it holds that
\[
(1-\tfrac{1}{2}\epsilon) \frac{1}{\ell n \theta} \leq P[e' = e] \leq \frac{1}{\ell n \theta}
\]
it follows that
\[
1-\tfrac{1}{2}\epsilon \leq \frac{P[e = e']}{P[e = e'']} \leq (1-\tfrac{1}{2}\epsilon)^{-1} \leq 1+\epsilon
\]
where the last inequality holds because $\epsilon \leq \frac{1}{2}$. This implies that the distribution is pointwise $\epsilon$-close to uniform.

\paragraph{Time complexity.}Since for every fixed edge $e'$, the probability that $e = e'$ is at least $(1-\epsilon) \frac{1}{\ell n \theta}$, the total success probability is
\[
P[\text{success}] = P[\bigvee e = e'] = \sum_{e' \in E} P[e = e'] \geq m (1-\epsilon) \frac{1}{\ell n \theta}
\]
where the second equality holds by disjointness of the events. The number of iterations of \Cref{alg:sample_edge} has geometric distribution with success probability of $m (1-\epsilon) \frac{1}{\ell n \theta}$. The expectated number of calls of \Cref{alg:sampling_attempt} is then
\[
\frac{\ell n \theta}{(1-\epsilon) m} = \frac{\sqrt{2} (\lceil \log_2 \frac{1}{\epsilon} \rceil+1) n}{(1-\epsilon) \sqrt{m}} = O\left(\frac{n}{\sqrt{m}} \log \frac{1}{\epsilon}\right)
\]
where the second equality is a substitution for $\theta$ and $\ell$.

\end{proof}

\section{Exact sampling} \label{sec:exact_sampling}
In this section, we describe an algorithm that samples an edge uniformly at random in wxpected time $O(\frac{n \log n}{\sqrt{m}})$. This algorithm works by using the approximate sampling algorithm (\Cref{alg:sample_edge}) with probability $1-\frac{1}{n^3}$ with $\epsilon = \frac{1}{n^3}$ and with probability $\frac{1}{n^3}$ a subrutine is called which corrects the deviations from the uniform distribution. This subrutine runs in time $O(m \log n)$ and it, therefore, does not influence the total expected running time. This approach can also be used with the algorithm from \cite{original_paper} but does not result in sublinear running time.

\bigskip \noindent
In order to correct the deviations from the uniform distribution, the algorithm has to know the probabilities with thich the individual edges are sampled -- that is, the values $h_{v,\ell}$ have to be computed. In fact, we compute the values of $h_{v,i}$ for all $v \in V, i \in [\ell]$. We do this iteratively in $\ell$ phases using the definition of $h_{v,i}$. Specifically, in the $i$-th phase, we have values of $h_{v,i-1}$ already computed, allowing us to compute the values of $h_{v,i}$. Each phase takes $O(m)$ time and there are $\ell = O(\log n)$ many of them, giving us a running time of $O(m \log n) \subseteq o(n^3)$. 

\bigskip \noindent
We now show the algorithm and argue its correctness.

\smallskip
\begin{center}
\begin{minipage}{.65 \textwidth}
\begin{algorithm}[H]
$h_{v,1} \leftarrow \frac{d_H(v)}{d(v)}$ for all $v$\\
\For{$i \in \{2,\cdots,\ell\}$}{
	$h_{v,i} \leftarrow h_{v,1} \sum_{w \in N_h(v)} h_{w,i-1} / d_H(v)$ for all $v$\\
}
With probability $1-\frac{1}{n^3}$, run \texttt{sample\_edge($1/n^3$)} and return its output\\
Otherwise, pick an edge at random, picking $e = (v,w)$ with probability~$\frac{1/m- (1-\epsilon) (1-h_{v,\ell})/\sum_{uw \in E} (1-h_{u,\ell})}{\epsilon}$\\

\caption{\texttt{Sample\_exactly}} \label{alg:sample_exactly}
\end{algorithm}
\end{minipage}
\end{center}
\smallskip

\medskip \noindent
Note that the sum of the probabilities is
\begin{gather}
\sum_{uw \in E}\frac{1/m- (1-\epsilon) (1-h_{v,\ell})/\sum_{u'w' \in E}(1-h_{u,\ell})}{\epsilon} =\\
= \frac{1 - (1-\epsilon) \sum_{uw \in E} (1-h_{u,\ell}) / \sum_{uw \in E} (1-h_{u,\ell}) }{\epsilon} = 1
\end{gather}
and the probabilities are non-negative:
\[
\frac{1/m- (1-\epsilon) (1-h_{v,\ell})/\sum_{uw \in E} (1-h_{u,\ell})}{\epsilon} \geq \frac{1/m- (1-\epsilon) /\sum_{uw \in E} (1-\epsilon)}{\epsilon} = \frac{1/m- 1/m}{\epsilon} = 0
\]
meaning that the probabilities indeed form a valid distribution.

\begin{theorem}
The \Cref{alg:sample_exactly} samples an edge from the uniform distribution (in expected time $O(\frac{n \log n}{\sqrt{m}})$.
\end{theorem}

\begin{proof}
Denote by $\mathcal{E}$ the event that the edge is sampled on line 1 and line 2 is, therefore, not executed. Let $e$ be the edge picked by the algorithm and $e'$ a fixed heavy edge. Note that $\epsilon = \frac{1}{n^2}$.
\begin{align}
P[e' = e] =& P[e'=e | \mathcal{E}] P[\mathcal{E}] + P[e'=e | \neg \mathcal{E}] P[ \neg \mathcal{E}] \\
=& \frac{1- h_{v,\ell}}{\sum_{uw \in E} (1-h_{u,\ell})}(1-\epsilon) + \frac{1/m- (1-\epsilon) (1-h_{v,\ell})/\sum_{uw \in E}(1-h_{u,\ell})}{\epsilon} \epsilon\\
=& \frac{1}{m}
\end{align}

\noindent
The expected time spent on line 1 is $O(\frac{n \log n}{\sqrt{m}})$. Since the probability of executing line 2 is $\frac{1}{n^3}$ in which case it uses $o(n^3)$ time, the expected time spent on line 2 is $o(1)$. Together, this shows the claimed bounds on the time complexity.
\end{proof}

\section{On the (un)importance of sampling exactly} \label{sec:doesnt_matter}
In this section, we give an argument that the success probability of an algorithm with expected running time $o(n^2) \subseteq o(n+m)$, changes by $o(1)$ when using approximate sampling with $\epsilon = \frac{1}{n^2}$. This follows from the existence of a coupling of approximate and exact edge sampling with small probability that the two samples from the coupled distributions are different.

\begin{theorem}
Given a graph $G$, let $\mathcal{O}_G$ be some subset of allowed outputs and let $\mathcal{A}$ be an algorithm in the extended sublinear graph model such that with probability at least $p$, $\mathcal{A}(G) \in \mathcal{O}_G$. Then for algorithm $\mathcal{A}'$, defined as $\mathcal{A}$ with every uniform random edge query replaced by an approximately (not necessarily pointwise) uniform random edge query with $\epsilon = \frac{1}{n^2}$, it holds that $P[\mathcal{A}'(G) \in \mathcal{O}_G] \geq p-o(1)$.
\end{theorem}

\begin{proof}
Let $\mathcal{U}$ be the uniform distribution on $E$ and $\mathcal{D}$ a distribution $\epsilon$-close to $\mathcal{U}$. By the coupling lemma \cite[Proposition~4.7]{levinmarkov}, there exists a coupling $(\mathcal{X},\mathcal{Y})$ of $\mathcal{D}$ and $\mathcal{U}$ such that $P[X \neq Y| X \sim \mathcal{X}, Y \sim \mathcal{Y}] = \|\mathcal{D}- \mathcal{U}\|_{TV} \leq \epsilon$. Let $\mathcal{A}_\mathcal{X}(G)$ and $\mathcal{A}'_\mathcal{Y}(G)$ be algorithms $\mathcal{A}, \mathcal{A}'$ sampling from the coupled distribution $(\mathcal{X}, \mathcal{Y})$, executed on input graph $G$. Note that $\mathcal{A}_\mathcal{X}(G) \sim \mathcal{A}(G)$ and $\mathcal{A}'_\mathcal{Y}(G) \sim \mathcal{A}'(G)$.

Let $k$ be the number of samples used by $\mathcal{A}$. Since $\mathcal{A}$ runs in expected time $o(n^2)$, and therefore, $k = o(n^2)$, it is the case that $P[(X_1, \cdots, X_k) \neq (Y_1, \cdots, Y_k)] = o(1)$ where $X_i,Y_i$ are the edges sampled according to $\mathcal{X},\mathcal{Y}$, respectively. This also means that $P[\mathcal{A}_\mathcal{X}(G) \neq \mathcal{A}'_\mathcal{Y}(G)] = o(1)$. Now

\[
P[\mathcal{A}'(G) \not \in \mathcal{O}_G] = P[\mathcal{A}_\mathcal{X}'(G) \not \in \mathcal{O}_G] \geq P[\mathcal{A}_\mathcal{Y}(G) \not \in \mathcal{O}_G] - P[\mathcal{A}_\mathcal{Y}(G) \neq \mathcal{A}'_\mathcal{X}(G)]
\]
from which the result follows since $P[\mathcal{A}_\mathcal{X}(G) \neq \mathcal{A}'_\mathcal{Y}(G)] = o(1)$ as we have argued. 
\end{proof}

\section{Acknowledgements}
I would like to thank my supervisor Mikkel Thorup for helpful discussions.

\bibliographystyle{plainnat}
\bibliography{references}
\end{document}